\documentclass[a4paper,11pt]{article}
\usepackage[english]{babel}
\usepackage[sectionbib]{natbib}
\usepackage[ansinew]{inputenc}
\usepackage{eurosym}
\usepackage{amsfonts}
\usepackage{mathrsfs}
\usepackage{booktabs}
\usepackage{amsthm}
\usepackage{amsmath, amssymb, bm}
\usepackage{mathtools}
\usepackage{graphics}
\usepackage[usenames,dvipsnames]{xcolor}
\usepackage{epsfig}
\usepackage{xcolor,colortbl}
\usepackage{dcolumn}
\usepackage{hyperref}

\usepackage[margin=10pt]{subfig}
\usepackage{pdfpages}
\usepackage[overlay]{textpos}
\usepackage{float}
\usepackage{xcolor}
\usepackage[xcolor=yap]{changes}
\usepackage{cancel}

\def\bbeta{\mbox{\boldmath $\beta$}}

\def\bdelta{\mbox{\boldmath $\delta$}} 

\def\bepsilon{\mbox{\boldmath $\epsilon$}} 
 \def\bmu{\mbox{\boldmath $\mu$}}
\def\bGamma{\mbox{\boldmath $\Gamma$}}

\def\bkappa{\mbox{\boldmath $\kappa$}}
\def\bpsi{\mbox{\boldmath $\psi$}}

\def\bSigma{\mathbf{\Sigma}}

  \def\bi{\mathbf{i}} 
\def\bE{\mathbf{E}} \def\bK{\mathbf{K}}

 \def\bw{\mathbf{w}} \def\bL{\mathbf{L}} 
\def\bu{\mathbf{u}}  \def\by{\mathbf{y}}  \def\0{\mbox{\bf{0}}}
\def\bS{\mathbf{S}}  \def\bI{\mathbf{I}}\def\be{\mathbf{e}}
\def\bZ{\mathbf{Z}}   \def\bM{\mathbf{M}}
    
 \def\bX{\mathbf{X}} \def\bx{\mathbf{x}}
 \def\bw{\mathbf{w}}%

             \def\WN{\mbox{WN}}
\def\E{\mbox{E}}       \def\diag{\mbox{diag}}
                                                        
                \def\vec{\mbox{vec}}

\def\nino{\mbox{Ni\~{n}o} }\def\nina{\mbox{Ni\~{n}a}}

\newtheorem{theorem}{Theorem}
\newtheorem{corollary}{Corollary}

\renewcommand{\baselinestretch}{1.0}
\begin{document}
\renewcommand{\baselinestretch}{1.2}
\title{\textbf{On the Estimation of \\ Climate Normals and Anomalies }  }
\author{Tommaso Proietti\footnote{Address for Correspondence: Dipartimento di Economia e Finanza, Via Columbia 2, 00133 Rome, Italy. E-mail: tommaso.proietti@uniroma2.it} \\  Universit\`{a} di  Roma ``Tor Vergata" \and Alessandro Giovannelli \\ Universit\`{a} dell'Aquila}
\date{ }
\maketitle
\begin{abstract}
The quantification of the interannual component of variability in climatological time series is essential for the assessment and prediction of the El Ni\~{n}o - Southern Oscillation phenomenon. This is achieved by estimating the deviation of a climate variable (e.g., temperature, pressure, precipitation, or wind strength) from its normal conditions, defined by its baseline level and seasonal patterns. Climate normals are currently estimated by simple arithmetic averages calculated over the most recent 30-year period ending in a year divisible by 10.
The suitability of the standard methodology has been questioned in the context of a changing climate, characterized by nonstationary conditions.
The literature has focused on the choice of the bandwidth and the ability to account for trends induced by climate change.
The paper contributes to the literature by proposing a regularized real time filter based on local trigonometric regression, optimizing the estimation bias-variance trade-off  in the presence of climate change, and by introducing a class of seasonal kernels enhancing the localization of the estimates of climate normals.
Application to sea surface temperature  series in the \nino  3.4 region and zonal and trade winds strength in the equatorial and tropical Pacific region, illustrates the relevance of our proposal.

\vspace{.5cm}
\renewcommand{\baselinestretch}{1}

\vspace{.5cm}
\end{abstract}
\thispagestyle{empty}
\renewcommand{\baselinestretch}{1.0}
\newpage

\section{Introduction }
\label{sec:intro}
Climate anomalies refer to deviations from baseline climatological conditions, including the seasonal cycle, and   quantify interannual climate variability on a year-to-year time scale. A dominant driver of this variability is the El  Ni\~{n}o-Southern Oscillation (ENSO) phenomenon. From a time series analysis perspective, the task of estimating anomalies can be framed as a band-pass filtering problem aimed at isolating the interannual variability. This variability excludes contributions from the long-term climate trend, and from  the seasonal cycle. In the ENSO case, the interannual component is typically defined to include frequencies corresponding to periods longer than one year, and shorter than approximately 5 to 7 years, with the upper periodicity being less rigidly constrained. 

The practices and guidelines for calculating climate normals (CN) are established by the World Meteorological Organization (WMO). As detailed in  \cite{wmo2017,wmo2018}, climate normals are defined as simple arithmetic averages of climate variables, such as temperature and precipitation, calculated over the most recent 30-year  period ending in a year divisible by 10 (e.g., from January 1, 1991, to December 31, 2020). For discussions on the scope and limitations of this definition, see  \cite{kunkel1990climatic}, while \cite{arguez2012noaa} and \cite{tveito2021norwegian} provide comprehensive overviews and discussion on practical implementations.

According to the \cite{wmo2017}, the 30-year period was chosen primarily for pragmatic reasons: ``the 30-year period of reference was set as a standard mainly because only 30 years of data were available for summarization when the recommendation was first made''. However, the concept of normality inherently presupposes the use of a sufficiently long reference period. Interestingly, as noted in \citet[][Section 2]{wmo2017},  climate normals serve a dual purpose: providing a reference value for the computation of climate anomalies and offering a baseline for predicting future climate conditions.

The U.S. National Oceanic and Atmospheric Administration (NOAA\footnote{https://www.ncei.noaa.gov/products/land-based-station/us-climate-normals}) also produces shorter-period normals, such as the   15-year normals. These shorter-term normals are particularly useful for sector-specific applications, such as energy load forecasting and construction planning. For the Oceanic Ni\~{n}o Index (ONI), climate normals are computed by averaging 30-year base periods centered on the beginning year of each 5-year period.

In sum, the standard methodology underlying climate normals typically involves the application of a fixed bandwidth, a moving reference period, and an equally weighted causal filter, which is an unweighted temporal average of three decades. This filter is updated either once per decade or every five years.

The suitability of the standard methodology has been questioned in the context of a changing climate, characterized by nonstationary conditions. If on the one hand updating the reference period accounts for time variation of climate normals, see e.g. \cite{NOAA2022}, on the other  the choice of the latter, along with the bandwidth and temporal distribution of the weights, have been under scrutiny.
For example, \cite{arguez2011definition} emphasize the importance of assigning greater weight to recent observations to account for potential trends. Specifically, under global warming, estimates of climate normals for temperature are prone to a cold bias.

The earlier literature critically examined the derivation of optimal climate normals (OCN).   \cite{kunkel1990climatic} reviewed these studies and analyzed how the choice of bandwidth influences optimal climate normals. \cite{huang1996long} formally defined OCN as averages over the most recent  years, where the number of years is selected to minimize the mean squared prediction error (MSE) for forecast horizons up to one year. In this framework, forecasts for a given month are derived from climate means for the corresponding season over the most recent $k$ years, with predictive accuracy assessed by the correlation between predicted and observed anomalies.

\cite{livezey2007estimation} argue that the WMO's traditional 30-year climate normals fail to adequately represent current climate conditions and lack predictive utility in the presence of linear trends in climate variables. They advocate for an alternative approach that involves fitting and extrapolating linear or piecewise linear trends, rather than relying on constant 30-year averages. Specifically, they propose using a two-phase regression model   incorporating a ``hinge" function, where the slope in the earlier period is constrained to zero. \cite{wilks2013performance} compare the predictive accuracy of various OCN and WMO methods.
Other approaches to the estimation of climate normals are based on local linear regression methods or cubic splines,  applied to annual time series relative to the months of the year, see, e.g., \cite{scherrer2024estimating} and the references therein.

The topic has great relevance since the analysis and the prediction of the ENSO phenomenon is based by and large on time series of anomalies, which enter predictive models such as vector autoregression and principal components analysis (aka  empirical orthogonal function analysis). Nonstationarities due to unaccounted climate trends and seasonal shifts will affect subsequent analysis by distorting both the serial dependence and the size of the anomalies.

The objective of this paper is to devise a (i) linear (ii) time invariant (iii) finite impulse response (FIR) (iv) causal (i.e., real time, or one-sided), filter for estimating monthly climate normals. The filter is based on a novel local trigonometric regression approach, that features a linear trend component along with trigonometric cycles defined at the seasonal frequencies. Its properties depend on the bandwidth parameter (number of most recent observations used in the estimation) and on a shrinkage parameter that regulates the bias-variance trade-off in the presence of a trend. Moreover, a class of seasonal kernels is introduced, that enforces the localization of the estimates.
The selection of the bandwidth and the shrinkage parameter is carried out by minimizing the mean square error incurred in the estimation of the climate normals.

The design of boundary filters is an important topic in seasonal adjustment; in the local polynomial regression framework \cite{proietti2008real} provided a solution for trend extraction filters aiming at minimizing the mean square error of the revisions. This paper proposes a different solution for the design of a local trigonometric filter with the above properties (i)-(iv), based on the minimization of the mean square estimation error. The latter takes explicitly into consideration the serial dependence of climate anomalies, which is  an additional contributions of the current paper, with respect to the OCN literature.

Our empirical illustrations deal with a large dataset of monthly gridded time series of sea surface temperatures (SST)  in the \nino 3.4 region and of zonal and trade wind strength in the equatorial and tropical Pacific. 
The arithmetic average of the SST anomaly series across the \nino 3.4 region is the most popular indicator of the oceanic component of ENSO.  El \nino refers to the periodic warming of waters near the South-American equatorial Pacific coast around Christmas and results in  positive anomalies in SST. The opposite phase is called La \nina, during which the sign of the anomalies is reversed.
The wind strength anomalies  are also an important component of the ocean-atmosphere interaction.
In normal conditions zonal winds  blow in east - west directions along the equator driving equatorial currents, thereby pushing warm waters to flow eastwards, while trade winds blow from tropical high pressure regions towards the equatorial low pressure zone in north-east to south-west (northeasterly trades) and in south-east to north-west directions (southeasterly trades).   This is strengthened during the La \nina \:phase, reinforcing the upwelling of cold water in the Eastern Pacific, during which positive wind strength anomalies emerge. Hence, in the wind series, a positive anomaly is a manifestation of the cold phase of ENSO. During a positive El \nino phase the winds weaken or even reverse, and  warm waters flow eastwards and this is evidenced by negative wind strength anomalies.
A thorough description of ENSO is found in \cite{neelin2010climate} and \cite{mcphaden2021nino}.

The structure of the paper is the following.
Section \ref{sec:opt}  addresses the optimality of the official climate normal  estimates: it shows that (a time invariant version of) the climate normal filter is optimal in a  local trigonometric regression framework, according to which a uniform kernel is assumed,  the conditional mean of the series features a constant level and a seasonal cycle, consisting of the fundamental frequency and all the harmonics, and the bandwidth is  fixed.
The optimality no longer holds if a linear trend is present, in which case the filter has a bias (Section \ref{sec:ltrtrend}); the optimal boundary filter in this situation is obtained by adding a trend adjustment filter to the original CN filter. The bias reduction comes at the expense of the estimation error variance and this opens the way to a regularization of the filter (Section \ref{sec:shrink}), that shrinks it towards the constant CN filter. In section \ref{sec:kernels} we address the question of defining a kernel for the local trigonometric regression model. Section \ref{sec:mse} discusses the selection of the bandwidth and the shrinkage parameter.
The estimation of climate normals and anomalies for SST and wind stress is considered in Section \ref{sec:data}. Section \ref{sec:concl} offers some conclusions.

\section{Optimality of Climate Normal Estimates}
  \label{sec:opt}

Let $\{y_t, t=1, 2, \ldots, n\},$ denote a climate time series observed in discrete time at a sub-annual frequency, and assume that it can be decomposed as
$y_t = \mu_t+\psi_t$, where $\mu_t$ denotes the climate normal at time $t$, and $\psi_t$ the corresponding anomaly. The component $\mu_t$ captures both the long run evolution of the variable and the seasonal cycle, with period equal to $s$ time units. For monthly data, $s=12$; for daily data $s=365$, after a correction for leap years (e.g., averaging the values of February 28 and 29). The anomaly is defined as the deviation of $y_t$ from normal conditions.

Figure \ref{fig:sstnoaa} displays the monthly sea surface temperature time series (degrees Celsius) in the Ni\~{n}o 3.4 region  from January 1950 to December 2024, along with the climate normals estimated by NOAA  (upper panel). The variability of the latter represents only a minor portion of the variability of the series, which dominated by the interannual component, shown in the bottom panel.
\begin{figure}[h]
\begin{center}
\includegraphics[width=1\textwidth]{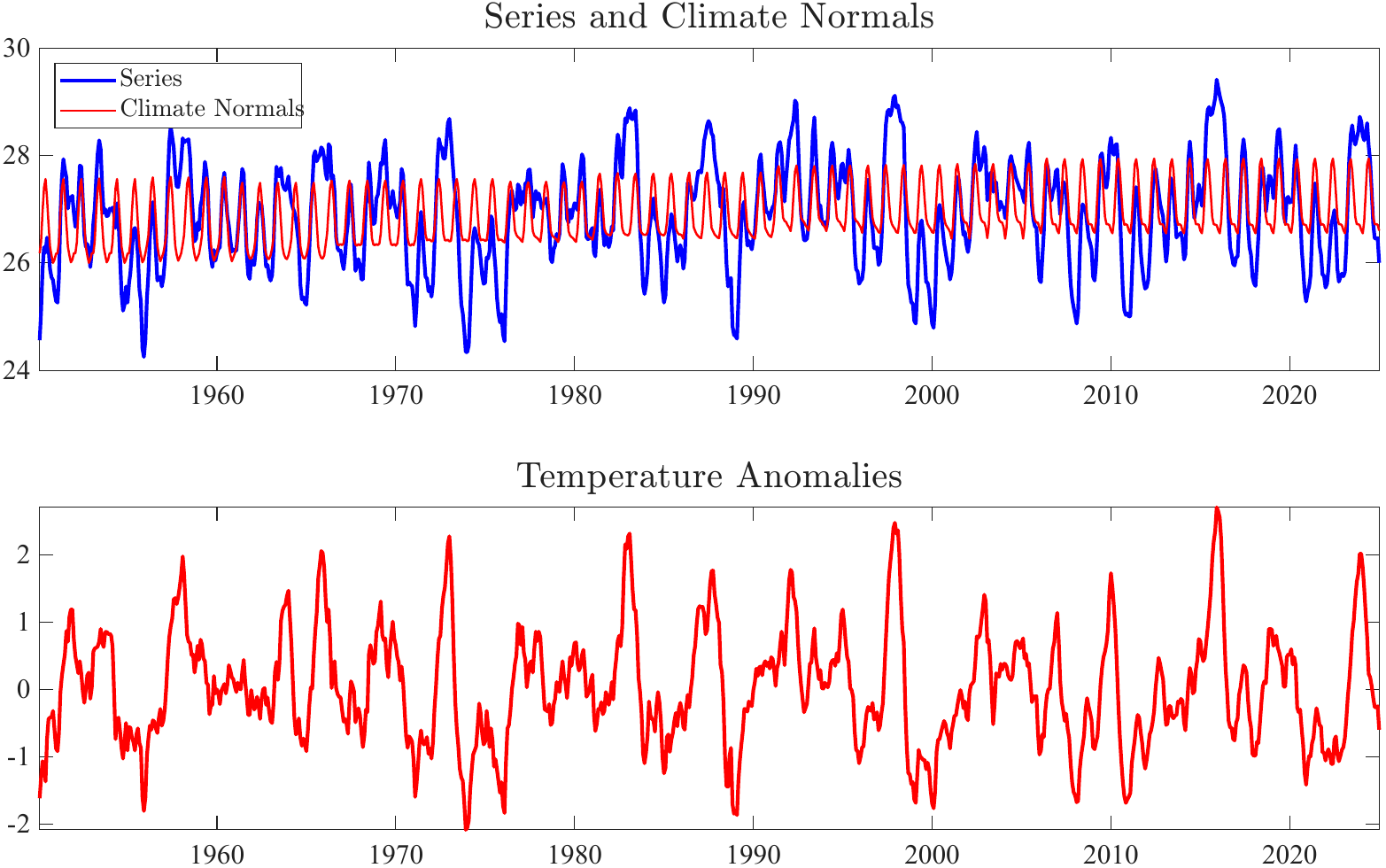}
\end{center}
\caption{Sea Surface Temperatures, Ni\~{n}o 3.4 region. NOAA climate normals series (top panel) and anomalies (bottom panel). \label{fig:sstnoaa}}
\end{figure}

The standard methodology adopts the estimator
\begin{equation}
\hat{\mu}_t^* = w_c(L^s)  L^{s(1 +\lfloor (t-t_0)/s\rfloor)}y_t, \;\;\; w_c(L) = \frac{1}{m+1}\sum_{j=0}^{m} L^j, \label{eq:w0}
\end{equation}
where $L$ is the lag operator, such that $L^k y_t=y_{t-k}$,   $t_0$ is the last season of the reference period, $\lfloor z \rfloor$ denotes the integer part of $z$, $w_c(L^s) = \frac{1}{m+1}(1+L^s+L^{2s} + \cdots+L^{sm})$ (the subscript $c$ stands for \emph{concurrent}).

For instance, the temperature climate normal of December 2024 is the arithmetic average of the $m+1=30$ December values in the years 1991-2020. It should be recalled that $t_0$ is updated every 10 (5) years at the beginning of a new decade (five year period).  The filter (\ref{eq:w0}) is time varying, since the lag operator  $L^{s(1 +\lfloor (t-t_0)/s\rfloor)}$ increases with $t$, modulo $10 s$.

What is the statistical model underpinning the estimator in (\ref{eq:w0})?
First and foremost, keeping the reference period fixed for a decade adds an artificial source of nonstationarity to the climate normals.
 A running mean estimator is more adequate, as it results from the application of a time invariant filter. Hence, in the sequel we discuss the optimality of the concurrent variant of (\ref{eq:w0}) which does not suffer from this limitation,
\begin{equation}
\hat{\mu}_t = w_c(L^s)y_t,  \;\;\; w_c(L^s)\frac{1}{m+1}(1+L^s+L^{2s} + \cdots+L^{sm}).
\label{eq:w1}
\end{equation}
To estimate trend and seasonality, the filter $w(L^s)$ computes the running mean of $m+1$ consecutive values of a specific season (month or day of the year). In other words, the monthly (daily) series $y_t$ is divided into 12 (365, after a correction for leap years) annual time series, and the climate normal is estimated by the running average of 30 consecutive values. 

An advantage of the fixed reference period estimator is that the climate normals need not be recomputed for the next five or ten years; however, this is a  minor argument as the computational costs of computing the running climate normals in real time is ignorable.
In general the properties of  a time invariant linear filter are better understood, and thus we will refer to it henceforth. Notice that the anomaly corresponding to (\ref{eq:w1}),
$\hat{\psi}_t = y_t -\hat{\mu}_t$, is 
$$\hat{\psi}_t = \frac{m}{m+1}\Delta_s y_t + \frac{m-1}{m+1}\Delta_s y_{t-s}+\cdots + \frac{2}{m+1}\Delta_s y_{t-s(m-1)}+\frac{1}{m+1}\Delta_s y_{t-sm},$$
where we have denoted $\Delta_s = 1- L^s$ (the seasonal differencing operator). Hence, the anomaly is estimated by  a moving average of $m$ year-on-year changes with linearly declining weights.

\subsection{Local trigonometric regression  }
 \label{ssec:ltr}

Let us turn to the question concerning the optimality of the moving average filter $w_c(L^s)$.  The filter can be interpreted as a boundary adaptation of the Nadaraya-Watson estimator of the climate normal of the season (e.g., month) occurring at time $t$, adopting a uniform kernel, see
\citet{fan1996local} and \citet{ruppert2003semiparametric}. 

The corresponding two-sided local level estimator of $\mu_t$ based on $q=2sm+1$ consecutive observations $y_{t+j}, j = 0, \pm 1, \ldots, \pm sm,$ around time $t$, is $w_q(L^s) = \frac{1}{2m+1}\sum_{j=-m}^{m} L^{sj}$.

For $s$ even, e.g. $s=12$, let us assume that in the neighbourhood of time $t$  the climate normals have the following representation,
\begin{equation}
\mu_{t+j} = \beta_0 + \sum_{k=1}^{s/2-1} \left\{\gamma_k \cos\left(\frac{2\pi j}{s}k \right) +\gamma_k^* \sin\left(\frac{2\pi j}{s}k \right)\right\}+\gamma_{s/2} \cos\left(\pi j\right), \label{eq:ltr}
\end{equation}
for $j = 0, \pm 1, \ldots, \pm  sm,$
while $\psi_t$ is a weakly stationary process with mean zero and variance $\sigma^2$. This implies that the value of the  climate normal at time $t$ is $\mu_{t} = \beta_0 + \sum_{k=1}^{s/2} \gamma_k$.
We refer to (\ref{eq:ltr}) as a local trigonometric regression (LTR) model, and the corresponding optimal signal extraction filter will be labelled LTR.
The model features  a level component (intercept), and $s/2$ trigonometric cycles defined at the fundamental frequency, $2\pi/s$ (one cycle per year), and the harmonic frequencies $2\pi k/s$ ($k$ cycles per year\footnote{If $s$ is odd the model  becomes
$$
\mu_{t+j} = \beta_0 + \sum_{k=1}^{ (s-1)/2 } \left\{\gamma_k \cos\left(\frac{2\pi j}{s}k \right) +\gamma_k^* \sin\left(\frac{2\pi j}{s}k \right)\right\},  j = 0, \pm 1, \ldots, \pm  sm.$$}), $k= 1, 2, \ldots, s/2$.

\begin{theorem}
Let $y_{t+j} = \mu_{t+j}+\psi_{t+j}$  for $j=0\pm 1, \pm sm$, where $\mu_{t+j}$ is given by \ref{eq:ltr} and $\psi_{t+j}$ is a zero mean  weakly stationary process. Then,  \begin{equation}
    \begin{array}{llll}
    \tilde{\mu}_t &=& \hat{\beta}_0+   \sum_{k=1}^{s/2} \hat{\gamma}_k & \\
                  &=& \sum\limits_{j=-sm}^{sm} w_{qj} y_{t-j},  & 
             w_{qj}  = \left\{\begin{array}{lr}\frac{1}{2m+1},  & j = 0,\pm s, \pm 2s, \ldots, \pm  sm,\\
                                                              & 0, \mbox{ otherwise.} \\
                                                              \end{array} \right.
    \end{array} \label{eq:cn2sided}
    \end{equation}
is the optimal least squares estimator of the climate normal at time $t$, based on a two-sided sample, and $\{w_{qj}, j =   0, \pm 1,\ldots, \pm sm\}$ is the impulse response function of the filter.

\end{theorem}
\begin{proof}


Focusing on the monthly case ($s=12$),  let $q = 24 m+1$, and define the $q\times 1$ vectors $\by = (y_{t+sm},\ldots,y_t,\ldots,y_{t-sm})'$, $\bpsi =  (\psi_{t+sm},   \ldots,   \psi_t, \ldots, \psi_{t-sm})'$, the $12\times 1$ vector $\bdelta = (\beta_0, \gamma_1, \gamma_1^*, \ldots, \gamma_5, \gamma_{5}^*, \gamma_6)'$, we write $\by = \bX\bdelta+\bpsi$, where $\bX = \left(\bX_f', \bx_0, \bX_p'\right)'$, with $$\bx_0' = (1, 1, 0, 1, 0, 1, 0, 1, 0, 1, 0, 1),$$ $\bX_p$ is $(sm)\times 12$ with $j$-th row equal to $$\left(1, \cos\frac{\pi}{6}j,\sin\frac{\pi}{6}j,\cos\frac{\pi}{3}j,\sin\frac{\pi}{3}j,\cos\frac{\pi}{2}j,\sin\frac{\pi}{2}j, \cos\frac{2\pi}{3}j,\sin\frac{2\pi}{3}j, \cos\frac{5\pi}{6}j,\sin\frac{5\pi}{6}j, \cos\pi j\right),$$
$j  = 1, 2, \ldots, sm,$
$\bX_f = \bE\bX_p\bS$, where $\bE$ is the $sm\times sm$ exchange matrix and $$\bS= \diag(1,1,-1,1,-1,1,-1,1,-1,1,-1,1)$$ is a matrix switching   the sign of the sine terms. These matrices have the following properties: $\bE'=\bE, \bE^2 = \bI_{sm}$, where $\bI_{sm}$ is the identity matrix of dimension $sm$, and $\bS\bS=\bI_{12}$.

The optimal filter weights   for estimating $\mu_t$ from the two-sided sample of $q$ observations $\by$, $\bw_q = \{w_{qj}, j=-sm, \ldots, -1,0, 1, \ldots, sm\}$, are the elements of the vector
$\bw_q = \bX(\bX'\bX)^{-1}\bx_0$, so that $\tilde{\mu}_t = \bw_q'\by$, or, equivalently,  $\tilde{\mu}_t = \bx_0'\hat{\bdelta}$,  with  $\hat{\bdelta} =(\bX'\bX)^{-1}\bX'\by$.

Since $\bX_p'\bX_p = \diag\left(sm, \frac{sm}{2}, \ldots, \frac{sm}{2}, sm\right)$, as implied by the orthogonality of the trigonometric functions, see  \citet[Theorem 3.1.1]{fuller2009introduction}, it follows that $\bX_f'\bX_f = \bX_p\bX_p$, and thus $\bX'\bX = 2 \bX_p'\bX_p+\bx_0\bx_0'$.

By the Sherman-Morrison-Woodbury inversion formula, see  \cite{henderson1981deriving},
$$(\bX'\bX)^{-1} = \frac{1}{2}(\bX_p'\bX_p)^{-1}-\frac{1}{4}\frac{(\bX_p'\bX_p)^{-1}\bx_0\bx_0'(\bX_p'\bX_p)^{-1}}{ 1+0.5 \bx_0'(\bX_p'\bX_p)^{-1}\bx_0}.$$

Hence, denoting $C_m = 0.5 \bx_0'(\bX_p'\bX_p)^{-1}\bx_0\equiv \frac{1}{2m}$,
\begin{equation}
\begin{array}{lll}
\bw_q' &=& \frac{1}{1+C_m} 0.5\bx_0'(\bX_p'\bX_p)^{-1} \left(\bX_f',\bx_0, \bX_p'\right)\\
    &=& \frac{1}{1+C_m} \left(0.5\bx_0'(\bX_p'\bX_p)^{-1}\bX_p'\bE,C_m, 0.5\bx_0'(\bX_p'\bX_p)^{-1}\bX_p'\right)\\
    &=&(\bw_p'\bE, \frac{1}{2m+1}, \bw_p'),
\end{array} \label{eq:w2}
\end{equation}
where  $\bw_p' = \frac{0.5}{1+C_m} \bx_0'(\bX_p'\bX_p)^{-1}\bX_p'$ has elements 
$$w_{pj} = \left\{\begin{array}{ll} \frac{1}{2m+1}, & \mbox{ if } j=12, 24, \ldots, 12m, \\
0, & \mbox{otherwise}.\end{array} \right.$$
This is so, since
$$ \frac{1}{2} + \sum_{k=1}^5 \cos\left(\frac{2\pi k}{12} j\right) + \frac{1}{2}\cos(\pi j)=\left\{\begin{array}{ll} 6, & \mbox{ if } j=0,12, 24, \ldots, 12m, \\
0, & \mbox{otherwise}.\end{array} \right.$$
In conclusion, the arithmetic moving average filter $w_q(L^s) = \frac{1}{2m+1}\sum_{j=-m}^{m} L^{sj}$, with impulse response $w_{qj} = \frac{1}{2m+1},$ $j = 0,\pm 12, \pm 24, \ldots, \pm 12m$, and $w_{qj} = 0$, otherwise, the optimal two-sided  signal extraction filter of the climate normals in (\ref{eq:ltr}).

\end{proof}

\begin{corollary}
The real time CN estimator, based on $(y_t, y_{t-1}, \ldots, y_{t-sm})'$,  and the corresponding least squares direct asymmetric filter (DAF) are, respectively,
        \begin{equation}
    \begin{array}{lll}
    \tilde{\mu}_{t|t}  &=& \sum\limits_{j=0}^{sm} w_{cj} y_{t-j},  \;\;\;
             w_{cj}  = \left\{\begin{array}{lr}\frac{1}{m+1},  & j = 0, s, 2s, \ldots,  sm,\\
                                                               0,& \mbox{ otherwise.} \\
                                                              \end{array} \right.
    \end{array} \label{eq:cn1sided}
    \end{equation}
\end{corollary}
\begin{proof}
Let $\bX_c = (\bx_0', \bX_p')'$,  where the subscript $c$ stands for `concurrent'. The DAF filter is obtained as $\bw_c = \bX_c(\bX_c'\bX_c)^{-1}\bx_0,$ giving the real time estimate $\tilde{\mu}_{t|t} = \bw_c'\by_c$, where $\by_c = (y_t, y_{t-1}, \ldots, y_{t-sm})'$.
Since $\bX_c'\bX_c = \bX_p'\bX_p+ \bx_0\bx_0'$, and applying the  Sherman-Morrison-Woodbury inversion formula, after some elementary manipulation, the direct asymmetric filter
can be written
$$\bw_c' = \frac{1}{1+\bx_0'(\bX_p'\bX_p)^{-1}\bx_0} \bx_0'(\bX_p'\bX_p)^{-1}(\bx_0, \bX_p'),$$
and can similarly be shown to have weights $w_{cj} = \frac{1}{m+1}$, for $j = 0, 12, 24, \ldots, 12 m,$ and $w_{cj} = 0$, otherwise.
\end{proof}


\section{Local Trigonometric Regression with a Linear Trend }
\label{sec:ltrtrend}
If a local linear trend is present, the climate normal filter (\ref{eq:cn1sided}) is no longer optimal. In particular, it will suffer from a bias component that depends on the size of the drift term.
Let the  model for climate normals be
\begin{equation}
\mu_{t+j} = \beta_0 + \beta_1 j + \sum_{k=1}^{s/2-1} \left\{\gamma_k \cos\left(\frac{2\pi j}{s}k \right) +\gamma_k^* \sin\left(\frac{2\pi j}{s}k \right)\right\}+\gamma_{s/2} \cos\left(\pi j\right),  \label{eq:ltrt}
\end{equation}
for $j = 0, \pm 1, \ldots, \pm  sm$,
so that, defining the variable $\mathbf{j}$, taking value $j$ at time $t+j$,  $\by =  \bX \bdelta + \beta_1 \mathbf{j} + \bpsi$.

We now show that the two sided filter $w_q(L^s) = \frac{1}{2m+1}\sum_{j=-m}^{m} L^{sj}$ is still optimal, but the real-time filter is different from $w(L^s)$. In particular,
the least squares real time (DAF) filter for estimating the climate normals is now obtained by correcting the equally weighted MA filter with an adjustment filter, giving
$$\tilde{\mu}_{t|t}^{(T)} = \sum_{j=0}^{sm} w_{cj}^{(T)} y_{t-j}, \;\;\; w_{cj}^{(T)} = w_{cj} + w_{aj},$$
where $w_{cj}  = \frac{1}{m+1},   j = 0, s, 2s, \ldots,  sm,$ and $w_{cj} = 0$, otherwise,
and the adjustment weights, $w_{aj}$ have an analytic formula provided below.

Including the trend does not alter the bi-directional filter: this can be seen by considering a reformulation of the local trigonometric model with linear trend, where $\mathbf{j}^*$ is replaced by $\mathbf{j}^* = \bM \mathbf{j}$, $\bM = \bI-\bX(\bX'\bX)^{-1} \bX'$, which is the residual from the projection onto $\bX$;
$\by = \bX \bdelta + \beta_1  \mathbf{j}^*+ \bepsilon$,
where  $\bepsilon = \bM\bpsi$.

It holds that $\tilde{\mu}_t = \bx_0'\hat{\bdelta},$ $\hat{\bdelta} = (\bX'\bX)^{-1}\bX'\by$, since $j_0^* = 0$, and due to the orthogonality of the regressors, the least squares estimate of $\bdelta$ is unchanged. Hence, the optimal two-sided filter is (\ref{eq:cn2sided}).

However, the direct asymmetric filter is different from $\bw_c$, as we show next.
Let  $\bM_c = \bI-\bX_c(\bX_c'\bX_c)^{-1} \bX_c'$ and $ \mathbf{j}_c = (0,1, \ldots, sm)'$.
Orthogonalizing the trend regressor and setting $\bM_c \mathbf{j}_c = \mathbf{j}^*_c$, the first element of the vector,  $j^*_{0c}$ is different from zero, and the optimal filter is
\begin{equation}
\begin{array}{lll}
\bw_c^{(T)} &=& \left( \bX_c,  \mathbf{j}_c^* \right)\left(\begin{array}{cc} \bX_c'\bX_c &  \0  \\ \0 & \mathbf{j}_c^{*'}\mathbf{j}_c^* \end{array}  \right)^{-1}
\left(\begin{array}{c} \bx_0     \\  u_{0c} \end{array}  \right)
\\
 &=& \bw_c + \bu_c(\mathbf{j}_c^{*'}\mathbf{j}_c^*)^{-1}j_{0c}^*\\
  &=& \bw_c + \bw_a,\\
 \end{array}
 \label{eq:cn1sidedt}
\end{equation}
where $\bw_c$ was derived above and   $\bw_a =  \mathbf{j}_c^*(\mathbf{j}_c^{*'}\mathbf{j}_c^*)^{-1}j_{0c}^*$ are the adjustment weights that account for the trend.
Obviously, $\bX_c'\mathbf{j}^*_c = \0$, and thus the adjustment weights sum to zero.
Hence, we have proven the following theorem.
\begin{theorem}
Let $\mu_{t+j}$ be generated according to (\ref{eq:ltrt}). The optimal (least squares) real time filter has impulse response function $\{ w_{cj}^{(T)}, j=0, 1, \ldots, sm\}$, with  elements resulting from the sum of $w_{cj}$, given  by (\ref{eq:cn1sided}), and the adjustment weights $\{w_{aj},  j=0, 1, \ldots, sm\}$, given by the element $j+1$  of the vector  $\bw_a$ in (\ref{eq:cn1sidedt}).
\end{theorem}

The adjustment weights are non-increasing in $j$. In conclusion, in the presence of a local linear trend, the arithmetic moving average filter commonly adopted is no longer optimal and the weights ought to be larger for observations closer to the current time.
\subsection{Regularized Local Trigonometric Filter}

 \label{sec:shrink}

The real time climate normals estimator (\ref{eq:cn1sidedt}) is unbiased if the model is correctly specified, but it suffers from a larger variance, since the leverage associated with the most recent observations is larger, due to the presence of a linear trend.

It can be more efficient to trade some bias for a reduction in the variance by shrinking the filter towards an equally weighted moving average filter, such as $\bw_c$.
The regularization of the filter can be performed  by introducing a scalar shrinkage intensity parameter, $\lambda \in [0,1]$,
to form the filter
\begin{equation}
\bw_r = \bw_c + \lambda \bw_a. \label{eq:regrz}
\end{equation}
The DAF in the presence of  a linear trend arises for $\lambda =1$ and the DAF for the no trend case  for $\lambda = 0$.

Figure \ref{fig:sk} plots the one sided filter weights $\{w_{cj}, j = 0, 1, \ldots, sm\}$ (top left panel), the trend adjustment weights $\{w_{aj}, j = 0, 1, \ldots, sm\}$ (top right) for $m = 10$ ($sm=120$). The bottom plots display the regularized filter weights $\{w_{rj}, j = 0, 1, \ldots, sm\}$ for $\lambda = 0.5$ (bottom left) and $\lambda = 1$ (bottom right).
\begin{figure}
\begin{center}
    \includegraphics[width=1\textwidth]{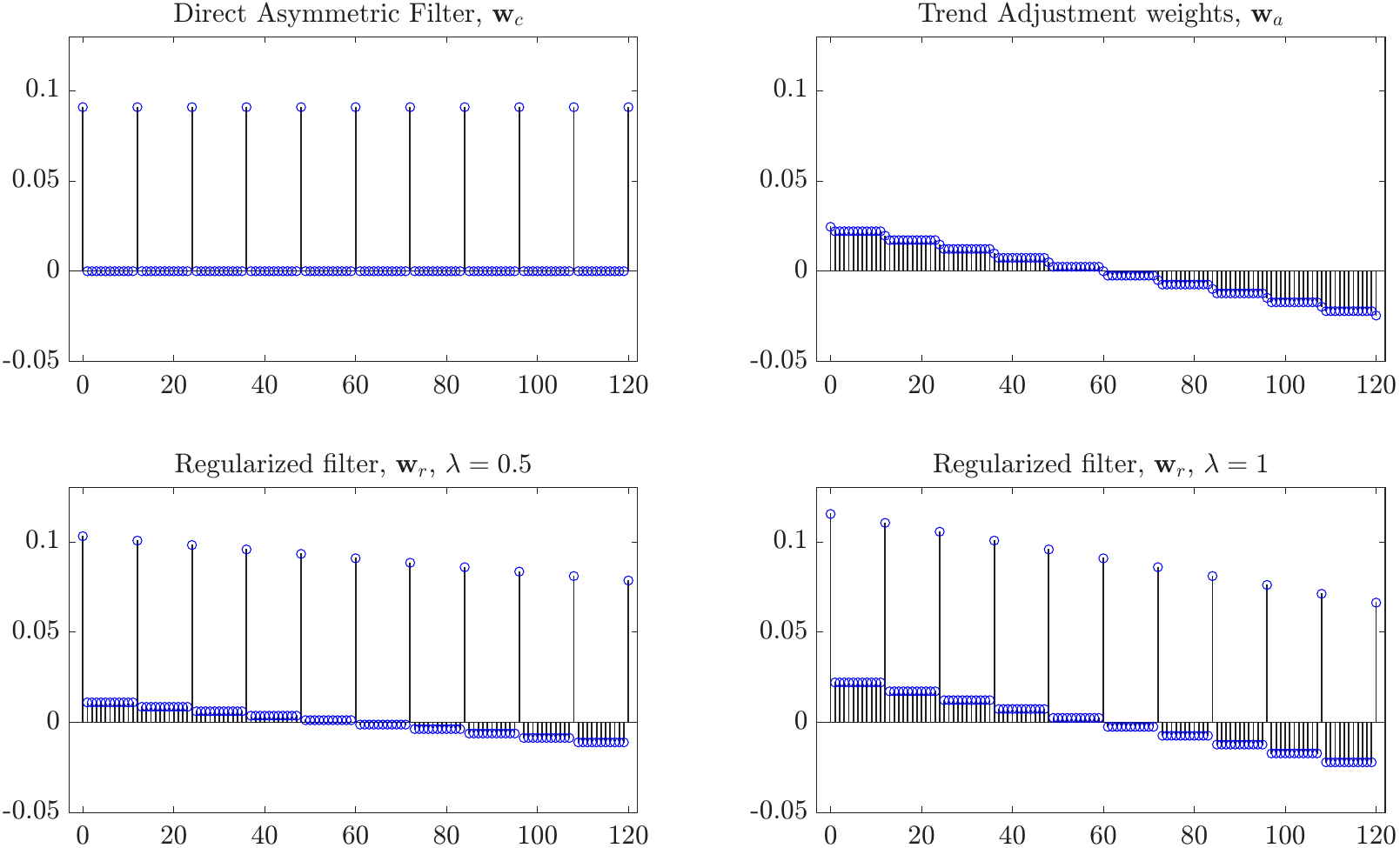}
    \caption{Climate normals one-sided filter weights with bandwidth $m=10$ and $s=12$:
    DAF weights (no trend) $\{w_{cj}, j = 0, 1, \ldots, sm\}$ (top left panel), trend adjustment weights, $\{w_{aj}, j = 0, 1, \ldots, sm\}$ (top right), regularized filter weights $\{w_{rj}, j = 0, 1, \ldots, sm\}$, for $\lambda = 0.5$ (bottom left) and $\lambda = 1$ (bottom right). In the last case the filter weights coincide with $w_{cj}^{(T)}$. \label{fig:sk}}
    \end{center}
\end{figure}

\section{Seasonal Kernels }
\label{sec:kernels}

    The above LTR filters adopt a uniform kernel. In the time series literature the arithmetic moving average filter $w_c(L)$  in (\ref{eq:w0}) is referred to as a  Macaulay filter.

In the local polynomial framework, when $\bX$ is $(2m+1) \times (p+1)$, with generic row $\bx_j' = (1, j , \ldots, j^p)'$, \cite{luati2011equivalence} show that  an  important class of candidate kernels, nesting the \citep{epanechnikov1969non} and the
\citep{henderson1916note} kernels,  arises  when it is assumed that the error is generated by the non-invertible moving average (MA) process of order $d$:
\begin{equation}
\psi_t = (1-L)^d \xi_t, \;\;\; \xi_t \sim \WN(0,\sigma^2).
\label{eq:nima}
\end{equation}
From the interpretative standpoint, the non-invertible MA  process in (\ref{eq:nima}) is the roughest stationary MA($r$) process, since its spectral density has a zero of order  $d$ at the zero frequency and increases monotonically from 0 to the Nyquist frequency. As a consequence, postulating this model amounts to impose a smoothness prior on the signal estimates.

Let us denote by $\bSigma_d$ the covariance matrix of the variables $\{\psi_{t-j}, \ldots, \psi_t, \ldots, \psi_{t+j}\}$  generated according to (\ref{eq:nima}). This is the symmetric $2m+1$-banded Toeplitz matrix, with the  coefficients associated with $L$ in the binomial expansion of $(1-L)^{2d}$,  displayed symmetrically about the diagonal in each row and column.
For instance, $\bSigma_1$ has elements $\sigma_{ii} = 2, \sigma_{i,i+1}= \sigma_{i-1,i} = -1,$ and $\sigma_{ij}= 0$ otherwise;
$\bSigma_2$ has elements $\sigma_{ii} = 6, \sigma_{i,i+1}= \sigma_{i-1,i} = -4, \sigma_{i,i+2} = \sigma_{i-2,i} = 1,$ and $\sigma_{ij}= 0$ otherwise.

The optimal kernel of order $d$, $\bkappa_d=(\kappa_{d,-m}, \ldots, \kappa_{d0}, \ldots, \kappa_{dm})'$, for the local polynomial regression model $\by = \bX\bbeta + \bpsi$, is such that the weighted least squares estimator of $\bbeta$,  $\hat{\bbeta} = (\bX'\bK_d\bX)^{-1}\bX'\bK_d\by$, $\bK_d = \diag(\kappa_{d,-m}, \ldots, \kappa_{d0}, \ldots, \kappa_{d,m})$, is equivalent to the GLS estimator, $\hat{\bbeta}_{GLS} = (\bX'\bSigma_d^{-1}\bX)^{-1}\bX'\bSigma_d\by$.

\cite{luati2011equivalence} show that the   kernel obtained from the row sums of the error precision matrix, $\bkappa_d = \bSigma_d^{-1}\bi$, enjoys this property if $p\leq d$, and that the kernel weights lie on a nonnegative polynomial of order $2d$:
\begin{equation}\kappa_{dj} \propto [(m+1)^2-j^2][(m+2)^2-j^2] \dots
[(m+d)^2-j^2],\label{eq:kq}\end{equation} for $j=-m,...,m.$
When $d=1$, $\psi_t=(1-L)\xi_t$, and $\bkappa_1$ is the Epanechnikov kernel, with elements $\kappa_{1,j} \propto
[(h+1)^2-j^2]$; the Henderson kernel, which can be seen as a discrete version of the triweight kernel, arises for $d=3$.  The uniform kernel arises for $d=0$, while  $d=2$ yields a discrete version of the biweight kernel.

More generally, it was proven in \cite{luati2011equivalence} that for arbitrary design matrix $\bX$ and covariance $\bSigma = \E(\bpsi\bpsi')$,
the optimal kernel that yields  the best  linear unbiased predictor of  $\by$ given $\bX$,   
is obtained as the solution of the homogeneous system  $\bL\bkappa = \0,$ where
\begin{equation}
   \;\;\;  \bL = \left(\bX'\otimes \left(\bSigma-\bX(\bX'\bSigma^{-1}\bX)^{-1}\bX'\right)\right)\bZ',
\label{eq:optK}
\end{equation}
and $\bZ$ the selection  matrix, such that $\vec(\bK)=\bZ' \bkappa$, and $\bZ\bZ' = \bI$.

This approach generalizes to the local trigonometric regression case, where $\bX$ is given in Sections \ref{ssec:ltr} and \ref{sec:ltrtrend}, as follows. The smoothness prior is now    $\psi_t = (1-L^s)^d \xi_t$, where $d>0$ and $\xi_t\sim\WN(0,\sigma^2)$; this implies that the spectrum of $\psi_t$ features zeros of order $d$ at the seasonal frequencies and at the long run frequency, i.e., all the variability at those frequencies is attributed to the climate normals. If $\bSigma_{d,s}$ denotes the autocovariance matrix of the process $\psi_t$, $\bSigma_{d,s} = \bSigma_d \otimes \be_1  \be_1'$, where $\be_1$ is an $s\times 1$ vector with 1 in first position and zero elsewhere, the corresponding kernel is $\bkappa = \bSigma_{d,s} \bi$.

The optimality follows from the fact that the null space of $\bL$ has dimension $s$ and that the kernel  $\bkappa = \bSigma_{d,s} \bi$ is a linear combination of the $s$ singular vectors spanning this space with $\bkappa\geq 0$.

Let $\bK= \diag(\kappa_{-sm}, \ldots, \kappa_0, \ldots, \kappa_{sm})$ be partitioned as
$\bK = \diag(\bK_f, \kappa_0, \bK_p),$ where $\bK_p =\diag(\kappa_1, \ldots, \kappa_{sm}),$ $\bK_f = \bE\bK_p$ and define  $\bK_c =\diag(\kappa_{0}, \bK_p).$
Then, the optimal two-sided filter is $\bw_q = \bK \bX (\bX'\bK\bX)^{-1}\bx_0$ and the DAF is $\bw_c = \bK_c \bX_c (\bX_c'\bK_c\bX_c)^{-1}\bx_0$.

Figure \ref{fig:sk} displays the seasonal Epanechnikov and Henderson kernel weights $\kappa_j$ for a half bandwidth of ten years of monthly observations, $m=10$, along with the corresponding optimal two-sided climate normals filters $\{w_j\}$  (bottom plots) for the trigonometric model (\ref{eq:ltr}) without a linear trend.
\begin{figure}
\begin{center}
    \includegraphics[width=1\textwidth]{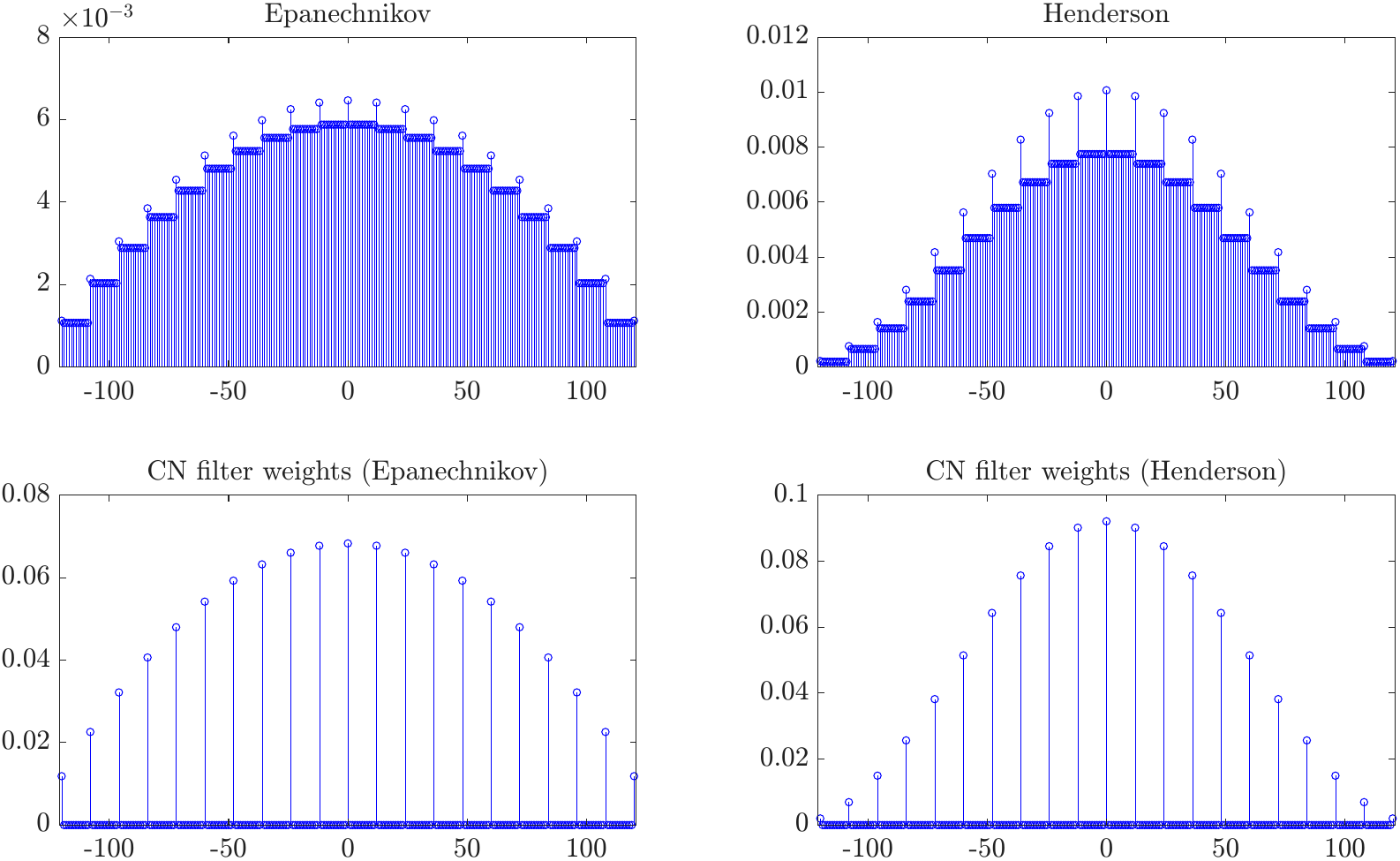}
    \caption{Epanechnikov and Henderson seasonal kernel weights, $m=10$, and optimal two-sided climate normal (CN) filter weights $\bw_q = \bK \bX (\bX'\bK\bX)^{-1}\bx_0$  (bottom plots) for the trigonometric model (\ref{eq:ltr}) without a linear trend. \label{fig:sk}}
    \end{center}
\end{figure}

For the trigonometric regression model with a linear trend (\ref{eq:ltrt}), the direct asymmetric filter is
$\bw_c^{(T)} = \bw_c+\bw_a$ with $\bw_a = \bK_c\bu_c(\mathbf{j}_c^{*'}\bK_c\mathbf{j}_c^*)^{-1}j_{0c}^*$. The regularized filter, $\bw_r$, is again constructed as in (\ref{eq:regrz}).
Figure \ref{fig:skf} displays the weights of the DAF filter for the no-trend case, $\{w_{cj}, j = 0,\ldots sm\}$ (top left), the trend adjustment weights  $\{w_{aj}, j = 0,\ldots sm\}$ (top right). The kernel is Epanechnikov's and $m=10$. The bottom plots show the weights of the regularized filter $\{w_{rj} = w_{cj}+\lambda w_{aj}, j = 0,\ldots sm\}$ for $\lambda=0.5$ and $\lambda =1$. The latter correspond to the DAF weights for the linear trend case.

\begin{figure}[h]
\begin{center}
    \includegraphics[width=1\textwidth]{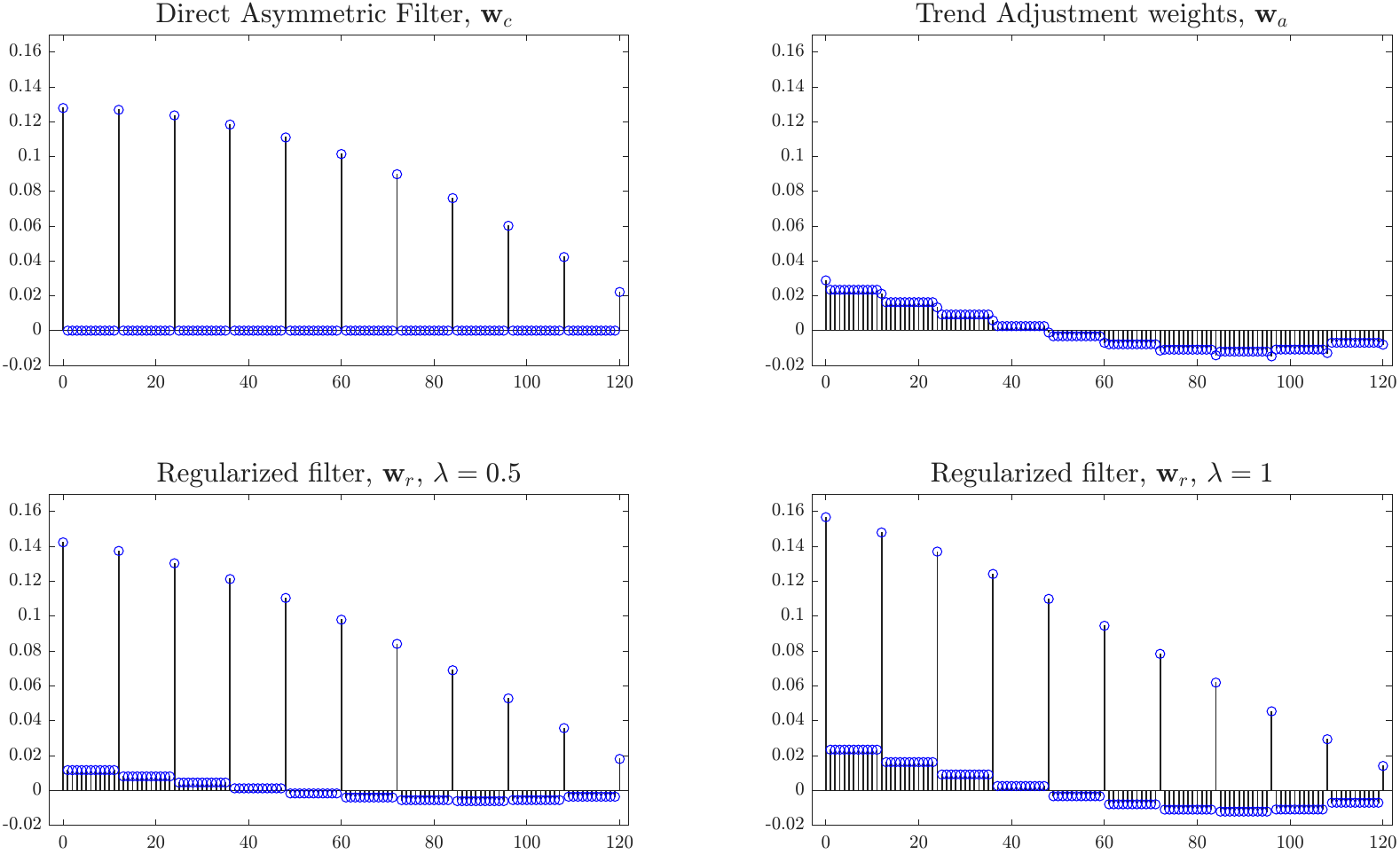}
    \caption{Climate normals one-sided filters using an Epanechnikov kernel with $m=10$: Weights of the DAF filter for the no-trend case, $\{w_{cj}, j = 0,\ldots sm\}$ (top left), the trend adjustment weights  $\{w_{aj}, j = 0,\ldots sm\}$ (top right).  weights of the regularized filter $\{w_{rj} = w_{cj}+\lambda w_{aj}, j = 0,\ldots sm\}$ for $\lambda=0.5$ and $\lambda =1$ (bottom plots). \label{fig:skf}}
    \end{center}
\end{figure}

\section{Selection of Bandwidth and Shrinkage intensity }
\label{sec:mse}

Given the choice of the kernel, the regularized climate normals estimator,
\begin{equation}\label{eq:cnest}
\tilde{\mu}_{t|t}^{(R)} = \sum_{j=0}^{sm} w_{rj}  y_{t-j}, \;\;\; w_{rj} = w_{cj} + \lambda w_{aj},
\end{equation}
depends on the choice of  $\lambda$ (shrinkage parameter) and $m$ (bandwidth in years). The optimality of the one-sided signal extraction filters for the LTR model (\ref{eq:ltrt}) are derived under the white noise or local stationarity assumptions for $\psi_t$. 

In the presence of a strong serially dependent  interannual component we need to prevent that the climate normals estimates absorb part or all the cyclical evolution of the interannual component. In the nonparametric literature, this is achieved by a suitable choice of the bandwidth, rather than through a modification of the signal extraction filters. The literature has investigated the optimal bandwidth choice in the in-fill asymptotics framework.  See  \cite{beran2002local} and \cite{beran2013long}, sec. 7.4, for overview of results. Time series cross-validation has been proposed, \cite{hart1991kernel}, see also \citet[Section 6.2]{fan2008nonlinear}.  
The climate literature has adopted time series cross-validation methods that ignore the interannual serial correlation, in that the choice of $m$
is based  on ability to predict the series by means of a misspecified model, such as (\ref{eq:ltrt}), assuming $\psi_t \sim \WN(0,\sigma^2)$.
%
%


We select the parameters that drive the properties of the CN filter, $(\lambda,m)$,  as the minimizers of the mean square error (MSE) of the CN estimator (\ref{eq:cnest}).
Denoting $\bmu = (\mu_t, \mu_{t-1}, \ldots, \mu_{t-sm})'$,  $\bpsi = (\psi_t, \psi_{t-1}, \ldots, \psi_{t-sm})'$, and recalling $\mu_t = \bx_0'\bdelta + \beta_1 j^*_{0c}$,
$\bmu = \bX_c\bdelta + \beta_1 \mathbf{j}^*$, $\bw_c'\bX_c = \bx_0'$, $\bw_a' \mathbf{j}^*_c = j_{0c}^*$,
the estimation error of $\tilde{\mu}_{t|t}^{(R)}$ is
$$\begin{array}{lll}
\tilde{\mu}_{t|t}^{(R)} - \mu_t & =& \sum_{j=0}^{sm} w_{rj} y_{t-j} - \mu_t \\
 & = & \sum_{j=0}^{sm} w_{rj} \mu_{t-j} - \mu_t + \sum_{j=0}^{sm} w_{rj} \psi_{t-j} \\
 & = & (\bw_r+\lambda \bw_a)'(\bX_c\bdelta +\beta_1 \mathbf{j}_c^*)-\bx_0'\bdelta - j_{0c}^*+\bw_r'\bpsi \\
 & = & (\lambda - 1)\beta_1 j_{0c}^*+ \bw_r'\bpsi,
 \end{array}$$
where the first addend contributes to the bias and the second to the variance.
Hence
\begin{equation}
\mbox{MSE}(\tilde{\mu}_{t|t}) = (\lambda-1)^2 \beta_1^2 j^{2*}_{0c} + \bw_r'\bGamma_{\psi}\bw_r, \label{eq:mse}
\end{equation}
where $\bGamma_{\psi}$ is the autocovariance matrix of $\bpsi$.

Both $\lambda$ and $m$ drive the bias-variance trade-off: having $0\leq \lambda < 1$ contributes the bias, but if $\lambda$ decreases, so does the variance, since the additional term $\lambda \bw_a'\bGamma_\psi\bw_a$ will make a smaller contribution.
As for $m$, while the variance is typically decreasing with $m$, the scalar $j^{2*}_{0c}$ increases monotonically with it.

%
%

\section{Estimation of Climate Normals and Anomalies for ENSO Characterization. }

 \label{sec:data}

\subsection{Description of the dataset}

We consider the estimation of the climate normals and anomalies for a set of 348 monthly time series available over a reference spatial grid, that are relevant for characterising the ENSO phenomenon. The time series are subdivided into two main groups. The first  relates to the oceanic component and consists of 120 series of Sea Surface Temperatures (SST)  in  the Ni\~{n}o 3.4 regions in the tropical pacific, while the second refers  to the atmospheric component of ENSO, namely zonal equatorial winds, and trade winds  (permanent east-to-west prevailing winds, also referred to as  `easterlies') located at different at different longitudes (West Pacific, Central Pacific, East Pacific).
Table \ref{tab:series} provides a summary description of the series considered in our illustrations.

The SST time series are extracted from the Extended Reconstructed Sea Surface Temperature (ERSST) dataset \citep{huang2017extended}, which contains global monthly sea surface temperature (SST) time series from 1854 to the present on a $2.5^{\circ}\times 2.5^{\circ}$ global grid. The dataset is available from the official NOAA repository at the website
\href{https://downloads.psl.noaa.gov/Datasets/noaa.ersst.v5/sst.mnmean.nc}{https://downloads.psl.noaa.gov}.
The series  are constructed from observations in the {International Comprehensive Ocean-Atmosphere Data Set (ICOADS)} \citep{freeman2017icoads}, combining ship-based measurements, buoy data, and Argo float records, as well as other sources depending on the area, see \cite{rayner2003global} for details.

\renewcommand{\baselinestretch}{1}
\begin{table}[hb]
\begin{center}
{\small
\begin{tabular}{lrlcc}
Series & $N$ & Region & Latitude & Longitude \\ \hline
    Sea Surface Temperatures & 120 & Ni\~{n}o 3.4 region & 5$^{\circ}$North-5$^{\circ}$South & 170$^{\circ}$West-120$^{\circ}$West \\
     ZWEq: 200 millibar (mb) Zonal Winds  & 23 & Equator & 0 &  165$^{\circ}$West-110$^{\circ}$West\\
     WPTWI: 850 mb Trade Wind Index & 95 & West Pacific & 5$^{\circ}$North-5$^{\circ}$ South & 135$^{\circ}$East-180$^{\circ}$West  \\
     CPTWI: 850 mb Trade Wind Index & 75 & Central Pacific & 5$^{\circ}$North-5$^{\circ}$ South & 175$^{\circ}$West-140$^{\circ}$West \\
     EPTWI: 850 mb Trade Wind Index & 35  & East Pacific    & 5$^{\circ}$North-5$^{\circ}$ South  & 135$^{\circ}$West-120$^{\circ}$West  \\\hline
\end{tabular}}
\end{center}
\caption{Description of the dataset. Distribution of the 348 monthly time series by group. The SST are available from January 1854 to December 2024. The winds series (meters per second (m/s)) are available from January 1948 to December 2024.  \label{tab:series}}
\end{table}
\renewcommand{\baselinestretch}{1}

The wind time series concern the wind velocity in the east-west direction, with positive anomalies indicating westerly winds   and negative anomalies measuring the strength of easterly winds. They are extracted   from the \textit{NCEP-NCAR Reanalysis 1} dataset,    documented in \cite{kalnay1996ncep} and downloadable at official NOAA repository \href{https://downloads.psl.noaa.gov/Datasets/ncep.reanalysis/Monthlies/pressure/uwnd.mon.mean.nc}{https://downloads.psl.noaa.gov}, which is provided at the monthly frequency starting  from January 1948 to the  present. The data are distributed on a global grid with standardized spatial resolution and interpolated into pressure levels, providing atmospheric dynamics at different altitudes.    This reanalysis provides a consistent global representation of atmospheric circulation and is based on a data assimilation system that incorporates observations from radiosondes, buoys, satellites, and numerical weather models.

\subsection{Preliminary analysis: stability of climate normals}

We preliminarily test the stability of the level and seasonal pattern (i.e., the components of the climate normals) of the series by means of the \cite{canova1995seasonal} statistic, modified by \cite{busetti2003seasonality}.  The null hypothesis is that seasonality and level are deterministic versus the alternative that they evolve according to a nonstationary stochastic process.

Let $e_t$ denote the ordinary least squares residual from regressing $y_t$ on an intercept  and 6 trigonometric components defined at the seasonal frequencies. The  statistic
$$\omega_j = \frac{1}{n^2\hat{\sigma}^2_L}\sum_{t=1}^n \left\{\left(\sum_{i=1}^t e_i \cos \varphi_j i\right)^2 +  \left(\sum_{i=1}^t e_i \sin \varphi_j i\right)^2 \right\},$$
provides a test for the null that the seasonal cycle at the frequency $\varphi_j = \pi j/6, j = 1,\ldots, 5,$ is deterministic,
where $\hat{\sigma}^2_L$ is a nonparametric estimate of the long run variance of $e_t$, for which we adopted a Bartlett window with bandwidth $M=12$ lags\footnote{Denoting the sample autocovariance of $e_t$ at lag $k$ by $\hat{\gamma}(k)$, $\hat{\sigma}^2_L = \hat{\gamma}(0) + 2 \sum_{k=1}^M \left(1-k/(M+1)\right) \hat{\gamma}(k)$. Above we set $M=12$.}.
The null distribution of the test is Cram\'{e}r-von Mises with 2 degrees of freedom; the $5\%$ critical value is 0.749 \cite[see][Table 1]{canova1995seasonal}.
For the zero (long-run) frequency and the $\pi$ frequency the test statistics are respectively
\begin{equation}
\omega_0 = \frac{1}{n^2\hat{\sigma}^2_L}\sum_{t=1}^n \left(\sum_{i=1}^t e_i \right)^2,  \;\;\;\omega_6 = \frac{1}{n^2\hat{\sigma}^2_L}\sum_{t=1}^n \left(\sum_{i=1}^t e_i \cos\pi i\right)^2. \label{eq:ch}
\end{equation}
Their 5\% critical value is 0.470. If a linear trend is included, then the distribution of $\omega_0$ is a level 2  Cram\'{e}r-von Mises with 1 degree  of freedom  and the critical value is 0.146.

Table \ref{tab:bhts} displays the deciles of the empirical distribution of the Busetti-Harvey statistics $\omega_j$, $j=0, \ldots, 6$, computed over the 120 SST series.
Stability at the long-run and seasonal frequencies is not rejected in almost all the cases.
\renewcommand{\baselinestretch}{1}
\begin{table}[hb]
\centering
\begin{tabular}{lccccccc}
 	&	$\omega_0$	&	$\omega_1$	&	$\omega_2$	&	$\omega_3$	&	$\omega_4$	&	$\omega_5$	&	$\omega_6$	\\\cline{2-8}
1st decile	&	0.039	&	0.017	&	0.004	&	0.004	&	0.001	&	0.001	&	0.002	\\
2nd decile	&	0.046	&	0.020	&	0.006	&	0.006	&	0.001	&	0.001	&	0.002	\\
3rd decile	&	0.051	&	0.022	&	0.007	&	0.006	&	0.001	&	0.002	&	0.003	\\
4th decile	&	0.056	&	0.024	&	0.007	&	0.007	&	0.001	&	0.002	&	0.004	\\
5th decile	&	0.060	&	0.025	&	0.010	&	0.008	&	0.001	&	0.002	&	0.004	\\
6th decile 	&	0.067	&	0.026	&	0.013	&	0.008	&	0.002	&	0.002	&	0.004	\\
7th decile	&	0.074	&	0.028	&	0.016	&	0.009	&	0.002	&	0.003	&	0.004	\\
8th decile	&	0.081	&	0.031	&	0.018	&	0.009	&	0.002	&	0.003	&	0.005	\\
9th decile	&	0.096	&	0.035	&	0.022	&	0.009	&	0.003	&	0.003	&	0.005	\\\cline{2-8}
\end{tabular}
\caption{\label{tab:bhts} Deciles of the distribution of the Busetti-Harvey stationarity test in the presence of trend over
$N=120$ Sea Surface Temperatures series in the Nino  3-4 region. For $j = 1,\ldots,5,$ the $5\%$ critical value  is 0.749; for $\omega_0$ it amounts to  0.146, while for $\omega_6$ it is equal to 0.470.
 }
\end{table}
\renewcommand{\baselinestretch}{1}

Very different conclusions are reached for the zonal and trade winds series.
Table \ref{tab:bhtw}, which  displays the deciles of the empirical distribution of the statistics $\omega_j$, $j=0, \ldots, 6$, computed over the 205 trade winds series, illustrates that trend stationarity is rejected in more than 90\% of the series. As far as seasonality is concerned, stationarity is rejected in at least 10\% of the cases only for the component at fundamental frequency ($\omega_1$, corresponding to one cycle per year), while it is not rejected at the harmonic frequencies.
\renewcommand{\baselinestretch}{1}
\begin{table}[hb]
\centering
\begin{tabular}{lccccccc}
        	&	$\omega_0$	&	$\omega_1$	&	$\omega_2$	&	$\omega_3$	&	$\omega_4$	&	$\omega_5$	&	$\omega_6$	\\ \cline{2-8}
1st decile	&	0.113	&	0.081	&	0.013	&	0.018	&	0.009	&	0.009	&	0.002	\\
2nd decile	&	0.165	&	0.130	&	0.024	&	0.026	&	0.015	&	0.012	&	0.003	\\
3rd decile	&	0.196	&	0.170	&	0.033	&	0.031	&	0.017	&	0.015	&	0.004	\\
4th decile	&	0.223	&	0.198	&	0.041	&	0.039	&	0.019	&	0.016	&	0.005	\\
5th decile	&	0.270	&	0.251	&	0.053	&	0.044	&	0.021	&	0.019	&	0.007	\\
6th decile 	&	0.315	&	0.315	&	0.065	&	0.051	&	0.023	&	0.021	&	0.008	\\
7th decile	&	0.331	&	0.410	&	0.075	&	0.059	&	0.029	&	0.023	&	0.013	\\
8th decile	&	0.358	&	0.576	&	0.109	&	0.075	&	0.034	&	0.026	&	0.019	\\
9th decile	&	0.413	&	0.931	&	0.192	&	0.109	&	0.047	&	0.031	&	0.027	\\\cline{2-8}
\end{tabular}
\caption{\label{tab:bhtw} Deciles of the distribution of the Busetti-Harvey stationarity test in the presence of trend over
$N=205$ trade winds series. For $j = 1,\ldots,5,$ the $5\%$ critical value  is 0.749; for $\omega_0$ it amounts to  0.146, while for $\omega_6$ it is equal to 0.470.
 }
\end{table}
\renewcommand{\baselinestretch}{1}

\subsection{SST and Winds Climate Normals and Anomalies}

Estimation of the climate normals and anomalies was conducted by choosing the pair of values $(\hat{\lambda}, \hat{m})$ that minimize the mean square estimation error, given in (\ref{eq:mse}). A grid search was conducted over the values of $\lambda$ from 0 to 1 with step 0.1 and considering all the values of $m$ from 6 to 30 with step 1. The regularized filter weights, $w_{rj} = w_{cj}+\lambda w_{aj}, j = 0, 1,  \ldots, sm,$  were obtained as described in Section \ref{sec:kernels}, using the one-sided seasonal Epanechnikov kernel. The adoption of a uniform kernel yields similar results, but the MSE is not a smooth function of the shrinkage and bandwidth parameters, and thus the estimates of $\lambda$ and $m$ tend to be  more erratic.

For the  SST series the average of the estimated values of $\lambda$ and $m$ resulted equal to 0.10 and 29.38, with standard deviation  0.07 and 0.93, respectively.
Figure \ref{fig:ssta} compares the cross-sectional average of the 120 estimated anomaly series across the Nin\~{o} 3.4 region
with the area-averaged anomaly series computed by the Climate Prediction Center (CPC) at NOAA, applying the 30-year filter \ref{eq:w0}, (reference period January 1991-December 2020, available at \href{https://www.cpc.ncep.noaa.gov/data/indices/}{www.cpc.ncep.noaa.gov}, series \texttt{detrend.nino34.ascii.txt}), plotted in Figure \ref{fig:sstnoaa}.  The differences are seasonal and seasonally varying, but their sizes do not exceed $0.3$ degrees Celsius. According to our analysis, the NOAA SST climate normals, already plotted in the upper panel of  Figure
 \ref{fig:sstnoaa}, when compared with the cross-sectional average of those estimated by the regularized LPR filter,
appear  adequate; in fact, as it was also pointed out by the seasonal stability tests, the evidence for varying trend and seasonal is rather weak.

\begin{figure}
\begin{center}
    \includegraphics[width=\textwidth]{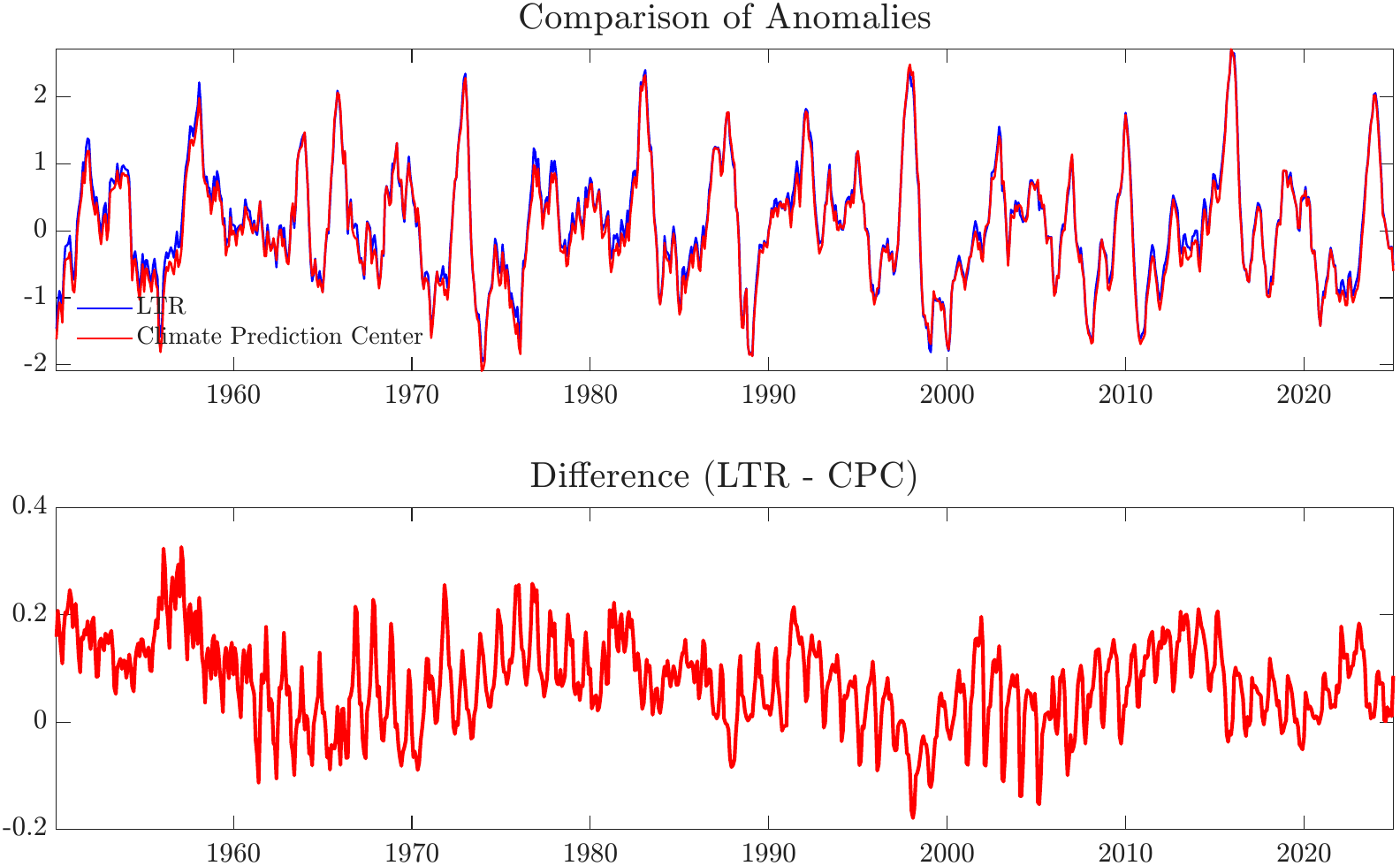}
     \end{center}
\caption{Comparison of SST anomalies: average of the 120 SST anomalies in the Nin\~{o} 3.4 region, estimated using the regularized local trigonometric regression (LTR) approach, versus the Climate Prediction Center (CPC) estimates. \label{fig:ssta}}
\end{figure}

The evidence is strikingly different for the 228 time series referring to zonal and trade winds. Table \ref{tab:winds} tabulates the joint distribution of the values of $(\hat{\lambda}, \hat{m})$ that minimize the mean square estimation error. The 30-year arithmetic average filter (\ref{eq:w1}), corresponding to $\lambda = 0$ and $m=30$  is selected only in 3 cases; the average of the estimates   are  equal to 0.49 and 23.86, with standard deviation  0.31 and 6.95, respectively for $\hat{\lambda}$ and $\hat{m}$.
\renewcommand{\baselinestretch}{1}
 \begin{table}
\centering
\small{\begin{tabular}{r|rrrrr|r}
   & \multicolumn{5}{c}{Values of $\hat{\lambda}$}\\
Values of $\hat{m}$ 	&	$0\dashv 0.2$	&	$0.2\dashv 0.4$	&	$0.4\dashv 0.6$	&	$0.6\dashv 0.8$	&	$0.8\dashv 1.0$	&	Total 	\\ \hline
$6-10$	&	0	&	2	&	12	&	0	&	0	&	14	\\
$11-15$	&	10	&	11	&	2	&	1	&	0	&	24	\\
$16-20$	&	13	&	16	&	1	&	1	&	0	&	31	\\
$21-25$	&	1	&	11	&	3	&	8	&	3	&	26	\\
$26-30$	&	37	&	13	&	25	&	19	&	39	&	133	\\\hline
Total	&	61	&	53	&	43	&	29	&	42	&	228	\\
\end{tabular}}
\caption{\label{tab:winds} Joint distribution of the values $(\hat{\lambda}, \hat{m})$ maximizing the mean square estimation error (\ref{eq:mse}). The counts  represent the number of time series for which the combination of  values $(\hat{\lambda}, \hat{m})$ in the corresponding ranges was selected.}
\end{table}
\renewcommand{\baselinestretch}{1}

Figure \ref{fig:twwa} compares the cross-sectional average of the $N=95$ West Pacific Trade Wind anomalies estimated by the regularized LTR filter with the   aggregate series \texttt{wpac850} published by NOAA's CPC at \href{https://www.cpc.ncep.noaa.gov/data/indices/}{www.cpc.ncep.noaa.gov}. The nonstationarity of normal climatological conditions, leaks into the anomalies. In particular, the 30-year filter tends to underestimate the climate normals, as it can be seen from Figure \ref{fig:twwcn}, and consequently to overestimate the anomalies. The difference in the anomalies tends to be negative and its size reaches up to 2 (m/s). As a result of unaccounted long-run variation, part of the persistence in the serial dependence of wind anomalies is due to the inability of the climate normals filter. The excess persistence of the wind anomalies is likely to affect subsequent joint modelling methods of SST and trade Winds by multivariate time series approaches, such as vector autoregressive models and principal components analysis.

\begin{figure}[h]
 \begin{center}
    \includegraphics[width=1\textwidth]{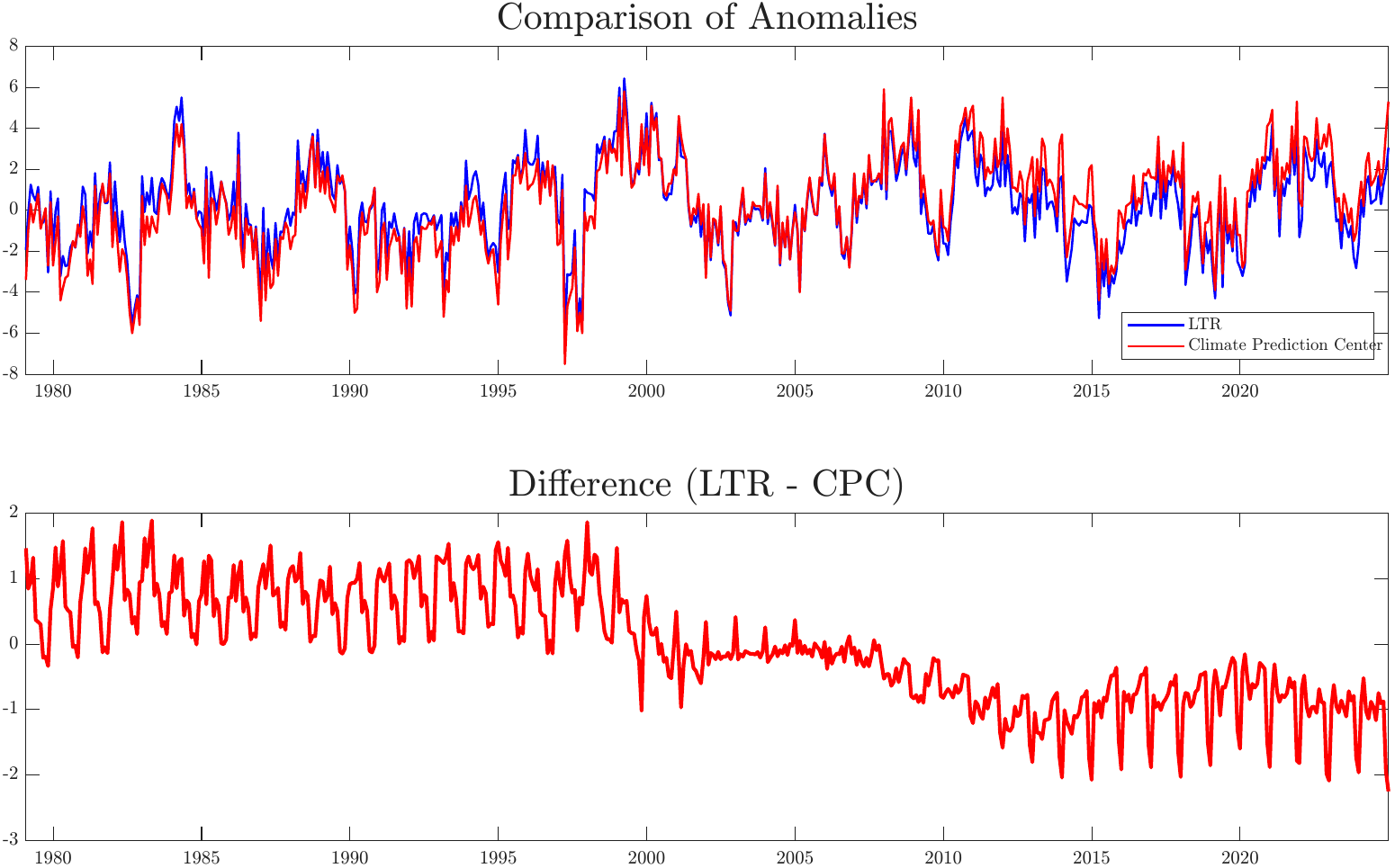}
     \end{center}
\caption{Comparison of  Trade Winds West Pacific anomalies: average of the 120 SST anomalies in the Nin\~{o} 3.4 region, estimated using the regularized local trigonometric regression (LTR) approach, versus Climate Prediction Center (CPC) estimates. \label{fig:twwa}}
\end{figure}

\begin{figure}
 \begin{center}
    \includegraphics[width=1\textwidth]{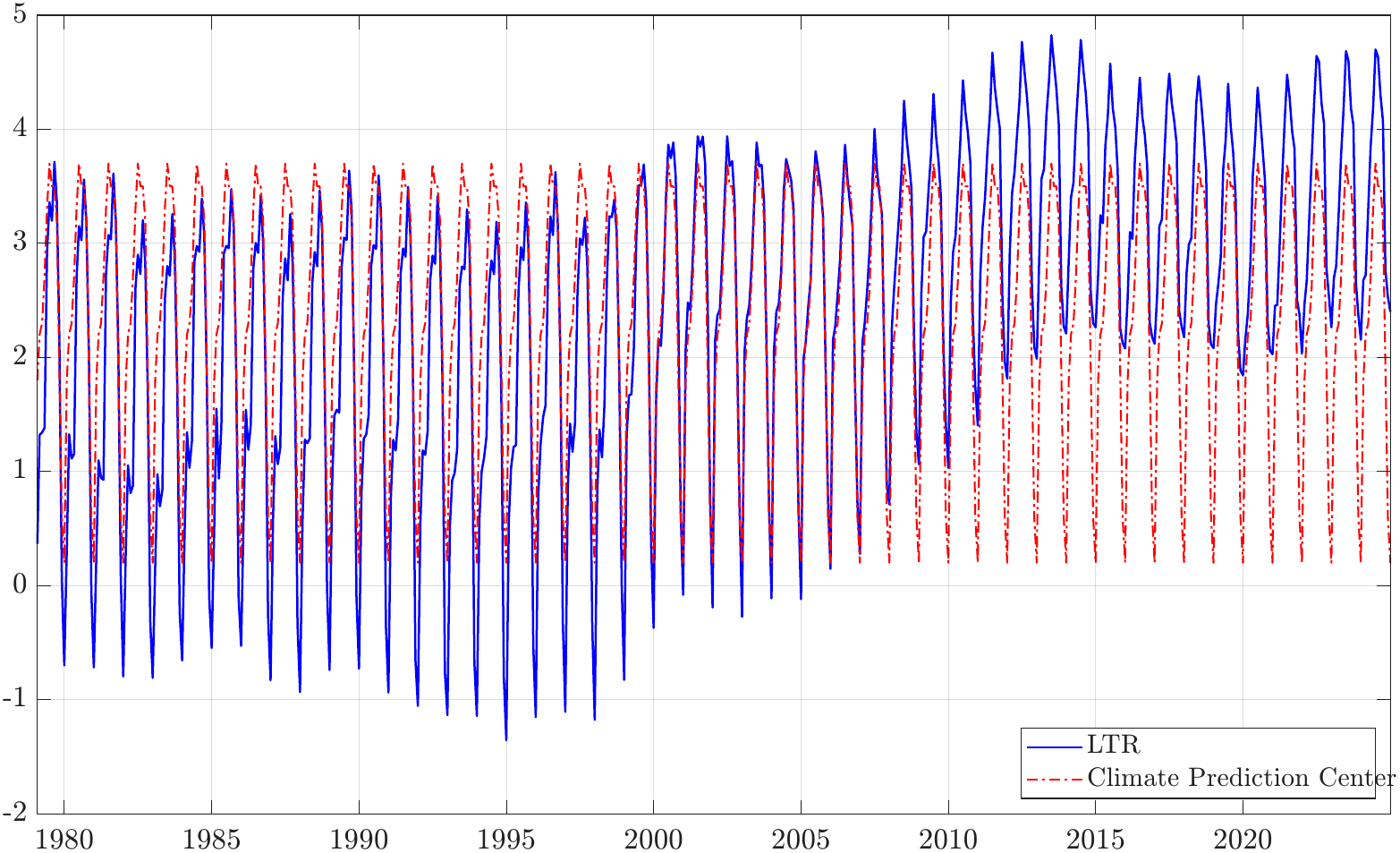}
     \end{center}
    \caption{Comparison of Trade Winds West climate normals: average of the 120 SST anomalies in the Nin\~{o} 3.4 region, estimated using the regularized local trigonometric regression (LTR) approach, versus Climate Prediction Center (CPC) estimates. \label{fig:twwcn}}
\end{figure}

Similar considerations hold for the other wind series in the dataset. Figure \ref{fig:cpac} compares the estimates of the climate normals and anomalies for the Central Pacific trade winds series for the period January 2010 - December 2024, obtained by the regularized LPR approach and by the CPC. The plot covers the period during which the climate normals estimated by the CPR are systematically below those estimated accounting for a potential trend.

It should be reminded that the differences are also due to the varying reference period, which in our case is continuously updated, while in the official estimates relates to the 3 decades 1991-2020. However,  the major differences are due to accounting for local trends ($\lambda > 0$), and to estimating, rather than fixing to 30 years the bandwidth (to 30 years).
The adoption of the kernel (Epanechnikov rather than uniform) plays a minor role for the anomaly estimates, but has an impact for the selection of the bandwidth and shrinkage intensity, in that the mean square estimation error is a smooth function of $\lambda$ and $m$.

\begin{figure}
\begin{center}
    \includegraphics[width=1\textwidth]{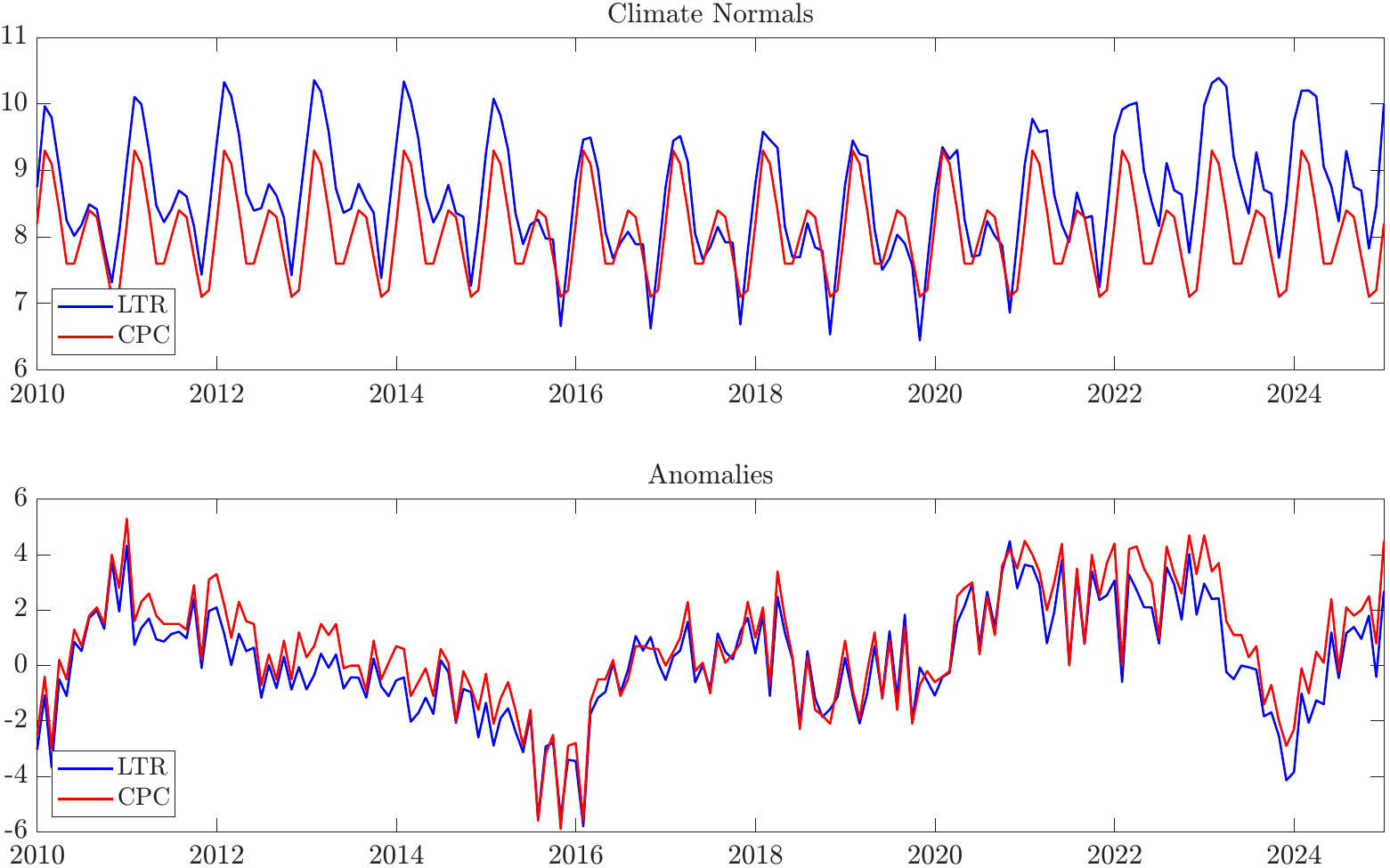}
     \end{center}
\caption{Comparison of Trade Winds Central Pacific climate normals and anomalies: average of LTR estimates versus Climate Prediction Center (CPC) estimate. \label{fig:cpac}}
\end{figure}



The wind stress anomalies estimated with the WMO filter (\ref{eq:w0}) are more persistent and possibly nonstationary. This is revealed by the distribution
of the statistic $\omega_0$ in (\ref{eq:ch}), where $e_t$ are the demeaned anomalies, over the 228 wind series. Table \ref{tab:winds2} reports the mean, median, the 95th percentile and the proportion of rejections of the null of stationarity at the 5\% significance level (i.e, the proportion of values that are larger than 0.47).
When the anomalies are estimated by the LRT approach, considering the minimun mean square estimates of  ($\lambda,m$), the evidence in favour of stationarity is much stronger, relative to the WMO anomalies. The rejection rate with $M=12$ is a mere 3.5\% as opposed to 55.3\%.
 \begin{table}
\centering
\small{\begin{tabular}{lrrrclrrr}
\multicolumn{4}{c}{Local Trigonometric Regression} & & \multicolumn{4}{c}{WMO method} \\
 	&	$M=6$	&	$M=9$	&	$M=12$	&	&	  &	$M=6$	&	$M=9$	&	$M=12$	\\
Mean	&	0.285	&	0.225	&	0.196	&	&	Mean	&	0.802	&	0.625	&	0.537	\\
Median	&	0.238	&	0.191	&	0.167	&	&	Median	&	0.773	&	0.603	&	0.522	\\
95th percentile	&	0.678	&	0.509	&	0.448	&	&	95th percentile	&	1.540	&	1.191	&	1.020	\\
Rejection rate	&	0.189	&	0.079	&	0.035	&	&	Rejection rate	&	0.697	&	0.636	&	0.553	\\
\end{tabular}}
\caption{\label{tab:winds2} Distribution (mean, median, 95th percentile and proportion of values greater than 0.47)  of the statistic $\omega_0$ for different values of the truncation parameter $M$ used for estimating the long run variance of the demeaned anomalies. }
\end{table}

\section{Concluding Remarks }
\label{sec:concl}
The measurement of climate normals and anomalies poses some interesting challenges, in the presence of climate change, such that both climate trend and seasonality are possibly time-varying. The study of anomalies is quintessential to the investigation of interannual climate variability (e.g. the ENSO phenomenon).
While we support the idea of applying a real time filter (whose properties are well understood), the latter should be capable of accounting for climate trends and evolving seasonality. Our contribution is to design enhancements of the standard climate normals filters that can achieve this task.
In particular, we have framed the problem of distilling climate anomalies into a local  trigonometric regression problem,  and have considered the selection of the three crucial ingredients, kernel, bandwidth and shrinkage intensity, by offering a selection criterion based on the minimization of the real time estimation error.

Our empirical illustrations, which motivated our proposal, led to conclude that while for the sea surface temperatures the WMO's traditional 30-year climate normals appear adequate, the wind series display varying trends and seasonality which require a different treatment. SST are indeed less affected by trend and seasonal permanent changes, and the climate normals are sufficiently stable to support the use of the signal extraction filters currently in use.
For both groups of series the interannual component is a major source of variability.

Though our treatment and case studies dealt with monthly time series, our approach easily extends to the daily frequency of observation.
Finally, we have only been concerned with one sided filters for real time estimation; in the previous sections the optimal two-sided filters were only relevant for framing the issue in a local trigonometric regression framework and for defining seasonal  kernels. This is not a serious limitation since the main use of anomalies is modelling and out-of-sample prediction, which require that the input series are measurable transformations of the information available at the reference time.



%
%
%
%

\end{document}